\documentclass[pra,superscriptaddress,twocolumn]{revtex4-2}
\usepackage{booktabs}
\usepackage{amsmath,amsthm,dsfont}
\usepackage{physics}
\usepackage{graphicx}
\usepackage[colorlinks=true,linkcolor=blue,citecolor=red,urlcolor=magenta]{hyperref}
\usepackage{tikz}
\newtheorem{theorem}{Theorem}

\newtheorem{proposition}{Proposition}

\usepackage{url}
\newcommand{\one}[0]{\mathds{1}}
\newcommand{\ot}[0]{\otimes}
\DeclareMathOperator{\sym}{sym}

\begin{document}

\title{The Richness of Bell Nonlocality: Generalized Bell Polygamy and Hyper-Polygamy}

\author{Gerard Angl\`es Munn\'e}
\affiliation{Institute of Theoretical Physics and Astrophysics, University of Gdańsk, 80-308 Gda\'nsk, Poland}

\author{Pawe{\l} Cie\'sli\'nski}
\affiliation{Institute of Theoretical Physics and Astrophysics, University of Gdańsk, 80-308 Gda\'nsk, Poland}
\affiliation{International Centre for Theory of Quantum Technologies, University of Gdansk, 80-309 Gda{\'n}sk, Poland}
\affiliation{Faculty of Physics, Ludwig Maximilian University, 80799 Munich, Germany}

\author{Jan W\'ojcik}
\affiliation{Institute of Theoretical Physics and Astrophysics, University of Gdańsk, 80-308 Gda\'nsk, Poland}

\author{Wies{\l}aw Laskowski}
\email{wieslaw.laskowski@ug.edu.pl}
\affiliation{Institute of Theoretical Physics and Astrophysics, University of Gdańsk, 80-308 Gda\'nsk, Poland}

\begin{abstract}
 Non-classical quantum correlations underpin both the foundations of quantum mechanics and modern quantum technologies. Among them, Bell nonlocality is a central example. For bipartite Bell inequalities, nonlocal correlations obey strict monogamy: a violation of one inequality precludes violations of other inequalities on the overlapping subsystems. In the multipartite setting, however, Bell nonlocality becomes inherently polygamous. This was previously shown for subsystems obtained by removing a single particle from an $N$-partite system. Here, we generalize this result to arbitrary $(N-k)$-partite subsystems. We demonstrate that a single $N$-qubit state can violate all $\binom{N}{k}$ relevant Bell inequalities simultaneously.  We further construct an $N$-qubit Bell inequality, obtained by symmetrizing the $(N-k)$-qubit ones, that is maximally violated by states exhibiting this generalized polygamy. We compare these violations with those achievable by GHZ states and show that polygamy offers an advantage in multipartite scenarios, providing new insights into scalable certification of non-classicality in quantum devices. Our analysis relies on symmetry properties of the MABK inequalities. Finally, we show that this behavior can occur across multiple subsystem sizes, a phenomenon we call hyper-polygamy. These structures reveal the remarkable abundance of nonlocality present in multipartite quantum states and offer perspectives for their applications in quantum technologies.
\end{abstract}

\maketitle

\section{Introduction}
    
Examining the behavior of non-classical correlations appearing in quantum mechanics has led to a deeper understanding of Nature's underlying structure. One of the most profound examples stems from Bell's theorem~\cite{Bell1964}, which made experimental tests of local realism possible. Bell nonlocality revealed through the violation of Bell inequalities is now not only a subject of purely foundational studies, but it is also one of the most important resources in the field of quantum technologies~\cite{Brunner_2014}. Its applications include quantum key distribution~\cite{Pawlowski_review_2025, Murta2020,Mayers_1998, Barrett_2005, Acin_2007, Pironio_2009}, randomness generation~\cite{Pironio_2010,Colbeck_2011,Colbeck_2011_2, ArnonFriedman_2018}, and communication~\cite{Cleve_1997,Brukner_2002, Brukner_2004,Moreno_2020} among others~\cite{Brunner_2014}. 

One of the fundamental principles of Bell nonlocality is its monogamy~\cite{SG2001, TVmono2006,Toner2009mono, Toner2007PhD}. It shows that the structure of nonlocal correlations in subsystems of a tripartite scenario with two settings per party is limited only to \emph{one pair of observers}. This finding has strong implications for communication security, see e.g.~\cite{Pawlowski_review_2025}. However, recently it was shown that the principle of monogamy holds universally only for the case of three observers and bipartite Bell inequalities. In the multipartite case, \emph{the nature of nonlocality is truly polygamous}. The findings reported in~\cite{poly} show that for a number of observers $N>3$ there always exists a $(N-1)$-qubit two-setting Bell inequality and an $N$-qubit state such that all $N$ inequalities on each $(N-1)$ partitions are violated simultaneously.

A phenomenon of monogamy violation was first observed in~\cite{SG2001} by examining the reduced states' nonlocal properties for a few overlapping parties with possibly different measurement settings.  Later, results for different Bell scenarios~\cite{Collins2004, Cui_2025} or few (but not all) parties~\cite{Kurzynski2011, Augusiak2014, Ramanathan2014, Tran2018, Augusiak2017, Ram_2018, Wiesniak2021} were reported, but~\cite{poly} provided a methodology for identifying the polygamous states and generating the desired inequalities for any number of observers in a two-setting two-output scenario.

In this work, we generalize this result to any $(N-k)$ party subsystems. We show that for any $k$ there exists a system size $N_k$ such that one can find a two-setting Bell inequality involving $N_k-k$ parties that is violated in all $\binom{N_k}{k}$ subsystems simultaneously. We call such a generalized Bell polygamous behavior $k$-polygamy. Our proof exploits the Mermin-type correlations arising from the Mermin-Ardehali-Belinskii-Klyshko (MABK) inequalities~\cite{MERMIN, Ardehali_1992, Belinskii1993} and linear programming for $N=6$, where the optimal $N_2$ is claimed. 
The same technique was used in~\cite{poly} to establish the optimal $N_1$.
Furthermore, we show that symmetrization of the $(N-k)$-qubit MABK inequalities is maximally violated by the states exhibiting the generalized polygamous behavior. We compare its violations for the chosen $k$-polygamous states with the standard approach using the GHZ state of $N-k$ qubits. 
We find that for $k=1$, 
the sum of the squared violation factors of all subsystems is the same in both cases. 
However, contrary to the GHZ states, in the
polygamous scenario, in all subsystems, we observe Bell violations. For $k>1$, the $k$-polygamous states perform better. 

At last, we observe that for a given positive integer $K$, witnessing $k$-polygamous correlations for all $1 \leq k\leq K$ at once is possible.
We named this behavior a $K$-hyper-polygamy.
This shows how vast the amount of Bell nonlocality in multipartite settings really is. 
Beyond standard applications of Bell nonlocality, such as communication complexity reduction (see e.g. ~\cite{Brukner_2002}), which can now be applied to multiple subsystems at once, our findings may have a potential impact on quantum cryptography protocols (conference key agreement~\cite{Murta2020}) or nonlocality-based benchmarking and certification of quantum devices.

\section{Generalized polygamy}

In ~\cite{poly}, it was shown that Bell nonlocality in multipartite systems is a polygamous phenomenon. More specifically, if the number of observers $N$ sharing a global qubit state $|\psi_N\rangle$ is greater than $3$, then all $N$ suitably chosen two-setting Bell inequalities on $(N-1)$-party subsystems can be violated \textit{simultaneously}. While $N=4$ requires a specific type of inequality, starting from $N=5$, such a behavior is present in the MABK (also referred to as Mermin-type) inequalities~\cite{MERMIN, Ardehali_1992, Belinskii1993}. Here, we extend this concept to any $(N-k)$-partite subsystems by proving that for all $k>0$ there exists a set of two-setting Bell inequalities on $(N-k)$ qubits and a global $N$-qubit state such that all $N \choose k$ inequalities are violated simultaneously.

We start by generalizing the notion of Bell polygamy from~\cite {poly} to an arbitrary number of discarded parties through the following. 
We state that if for any $k>0$, one can find a number of observers $N_{k}$ such that there exist a $N_k$-qubit state $|\psi_{N_k}\rangle$ and a two-setting Bell inequality $I_S<L_{N_k-k}$ with $S$ being a subsystem of size $N_k-k$ that
\begin{equation}\label{eq:kpoly}
     \forall S \,\,\bra{\psi_{N_k}} (I_{S} \otimes \one_k) \ket{\psi_{N_k}} > L_{N-k} \,,
    \end{equation}
the $k$-polygamy holds. 
Here, $\mathds{1}_k$ is a $k \times k$ identity matrix.
Thus, this system of $N_k$ observers simultaneously violates all inequalities $I_S$ for all $S$.

This also means that one can construct quantum states whose nonlocal properties are robust under the loss of arbitrary particles (even deterministically lost ones). Hence, Bell nonlocality shared between the observers and the number of inequalities violated by a suitable state can be almost arbitrary, and polygamous correlations introduced in~\cite{poly} are even more abundant in quantum mechanics. 

We will prove that the above holds for the MABK inequalities through the following theorem.
\begin{theorem}
    For any $k>0$ there exists a number of observers $N_k$ such that a violation of all $N_k \choose k$ MABK inequalities $\mathcal{M}_{N_k-k}<L_{N_k-k}$ among $(N_k-k)$ parties is possible. 
\end{theorem}
\noindent Note that the obtained $N_k$ does not have to be optimal in this case. Namely, there may exist other inequalities that, for a given $k$, yield a smaller $N_k$. However, a single inequality is sufficient to prove our claim. Later, we will discuss the possible improvement on this result using linear programming. 
Our approach to the studied problem starts with considering the general pure $N$-qubit permutation invariant state
\begin{equation}
    |\text{PI}_N\rangle =\sum_{i=0}^N c_i |D^i_N\rangle \,,
    \label{eq:PI_state}
\end{equation}
where $|D^i_N\rangle$ are the $N$-qubit Dicke states with $i$ excitations. The reason for examining a permutation invariant state is that its reduced states' expectation values on permutation invariant operators are equal, and choosing a Bell inequality with such a structure automatically yields the polygamous behavior. The $(N-k)$-party Mermin-type operator for optimal GHZ settings of the MABK inequalities can be expressed using the Dicke states in the following way~\cite{Belinskii1993, Scarani_2001, Nagata_2006}
\begin{eqnarray}\label{eq:merm_simp}
    M_{N-k}=2^{N-k-1}\left(|D^0_{N-k}\rangle \langle D^{N-k}_{N-k}| + |D^{N-k}_{N-k}\rangle \langle D^0_{N-k}|\right).\,\,
\end{eqnarray}
Any local and realistic model has to yield $M_{N-k}<L_{N-k}$, where $L_{N-k}=2^{(N-k-1)/2}$. With straightforward calculations, one can show that its expectation value on the global state $|\text{PI}_N\rangle$ is given as
\begin{eqnarray}\label{eq:merm_simply}
     \mathcal{M}_{N-k}  &:=& \bra{\text{PI}_N}(M_{N-k } \otimes \mathds{1}_k)\ket{\text{PI}_N} \nonumber \\
     &=& 2^{N - k} \sum_{s=0}^k {k \choose s} {N \choose s}^{-1/2} {N \choose k-s}^{-1/2} c_s \, c_{N-k+s} \nonumber\\
     &=& \sum_{s=0}^k A^{N,k}(s) \, c_s \, c_{N-k+s}.
\end{eqnarray}
Now, we shall prove that for $k<N/2$ the maximum quantum value of the considered inequality is given by
\begin{align}
     \mathcal{M}^{max}_{N-k} &= \max_s \frac{1}{2}A^{N,k}(s)=\frac{1}{2}A^{N,k}([\,k/2\,]),
     \label{eq:maxMABK}
\end{align}
where by $[\, \cdot \,]$ we denote the closest integer in both directions. Later, we will use this result to show when the above exceeds the local bound and thus leads to the claimed $k$-polygamous behavior. First, let us show that $\mathcal{M}^{max}_{N-k}=\max_s \frac{1}{2}A^{N,k}(s)$. 
When $k<N/2$, the variables $c_s$ and $c_{N-k+s}$ are independent and without loss of generality, we can assume that $c_s=0$ for $k\leq s\leq N-k$ since they do not contribute to the expectation value.
Due to state normalisation, the non-zero variables satisfy
\begin{align}
\sum^k_{s=0} (c^2_s + c^2_{N-k+s} )=\sum^k_{s=0} p_s =1 \,,
\end{align}
where $p_s:=(c^2_s + c^2_{N-k+s})\geq 0$.
Note that $c_sc_{N-k+s}$ appearing in Eq.~\eqref{eq:merm_simply} reaches its maximum value when $c_sc_{N-k+s}=p_s/2$. Thus, for a fixed $N$ and $k$, maximisation of Eq.~\eqref{eq:merm_simply} reduces to
\begin{align}
     \mathcal{M}^{\text{max}}_{N-k} = \max_s \frac{1}{2} \sum_{s=0}^k A^{N,k}(s) p_s = \max_s \frac{1}{2} A^{N,k}(s),
\end{align}
since it is a convex combination. Now, it is left to prove that the maximum of $A^{N,k}(s)$ is given for $s=[\,k/2\,]$. This can be seen by noting that $A^{N,\,k}(s)$ is symmetric with respect to $k/2$ and increases strictly up to this point. To verify this, let us show that for a given integer $s$
\begin{align}\label{eq:Aindec}
\frac{A^{N,\,k}(s)}{A^{N,\,k}(s+1)}=\frac{(N-s) (s+1)}{(N-(k-s)+1) (s)} \leq 1,
\end{align}
holds. By rearranging the terms, one can obtain the equivalent statement given as
\begin{align}
(N-k)(k-2s-1)\geq 0 \,.
\end{align}
The above holds for $0 \leq s \leq [\, k/2 \, ]$, and thus $A^{N,k}(s)$ increases up to $[k/2]$. Similarly, the function decreases from $s=k/2$ when $k$ is even, but from $s= \lceil k/2 \rceil$ when $k$ is odd. 
As a consequence, the $s$ corresponding to the maximum is unique for even $k$, but for odd, the maximum is reached by two values: $s=\lfloor k/2 \rfloor, \lceil  k/2  \rceil.$
This proves $\mathcal{M}^{\text{max}}_{N-k} =\frac{1}{2}A^{N,k}([\,k/2\,])$. 

The optimal value of $s$ determines the amplitudes $c_i$ and thus yields the optimal state that reaches the $\mathcal{M}^{\max}_{N-k}$. Explicitly, for a given $N$ and $k<N/2$ they are given as
\begin{align}\label{eq:max_poly}
       |\phi^k_N\rangle=\frac{1}{\sqrt{2}} \left( |D_N^{[\,k/2\,]}\rangle+ |D_N^{N-k+[\,k/2\,] }\rangle\right).
\end{align}
For $k$ odd, Eq.~\eqref{eq:max_poly} denotes two possibilities since there are two states reaching the maximum: when $[\,k/2\,]$ corresponds to $\lfloor k/2 \rfloor$ and to $\lceil k/2 \rceil.$
In fact, any superposition of these two optimal states will also maximize the expectation value.
For $k=1$ and taking $[ \,k/2\, ] =1$, the optimal state becomes 
\begin{eqnarray} \label{eq:1-poly}
    |\phi^{1}_{N}\rangle&=&\frac{1}{\sqrt{2}} \left( |D_N^{1}\rangle+ |D_N^{N}\rangle\right) \nonumber\\
    &=& \frac{1}{\sqrt{2}} \left( |W_N\rangle+ |1\cdots 1 \rangle\right),
\end{eqnarray}
where $|D_N^1\rangle  \equiv |W_N \rangle$ is the $N$-qubit W state, which is exactly the result obtained in~\cite{poly}.

The above result can be used to establish that for all considered values of $k$ there exists such $N_k$ that all MABK inequalities among $(N_k-k)$ parties are violated at once, i.e. $\mathcal{M}^{max}_{N_k-k}/L_{N_k-k}>1$. This follows from the fact that the violation factor for a given $k$ increases as a function of $N$ for high enough $N$, and it is unbounded. For example, in the case of even $k$, the violation factor increases as long as $N> (2 + \sqrt{2})k/2-1$ and then reaches infinity in the limit of $N\rightarrow \infty$. Hence, it has to exceed $1$ for some $N_k$, and this proves our Theorem. From the symmetry, $\mathcal{M}^{\text{max}}_{N-k}$ is the maximal expectation value on all ${N \choose k}$ of the subsystems involving $N-k$ parties. Since we have shown explicitly that there exists at least one family of inequalities and states for which $k$-polygamy holds, a general claim about the existence of the generalized polygamy has to hold. However, again, we note that the obtained $N_k$ for the MABK inequalities do not have to be optimal. This was already noted in~\cite{poly} for $k=1$.

One interesting fact about the generalized polygamy is that if for a given $N$ one can violate all $(N-k)$-party Bell inequalities, then with a different state, all $(N-k')$-party Bell inequalities with $k'<k$ can also be violated. This is clear from the monotonicity of $\mathcal{M}^{max}_{N-k}/L$, and makes the minimal $N_k$ the number of interest. A few initial values of the minimal $N_k$ for which the generalized polygamy holds are presented in Table~\ref{tab:poly}. Notably, the minimal $N_k$ for the MABK inequalities scales approximately as a linear function of $k$. For $N\leq 500$, we observe it to follow $a \cdot k+b$, where $a \approx 2.77844$ and $b \approx 1.21748$. Using it \textit{only two} errors after rounding to the nearest integer of $\pm 1$ appear (for $k=1$ and $k=274$), and hence this function approximates the desired minimal $N_k$ well. Based on the interpolation, the maximal number of parties that can be discarded from each subsystem is roughly $N/3$.

In the presence of an additional source of noise, the provided $N_k$ values can increase. Consider the depolarizing noise that transforms the initial $N$-qubit state $\rho$ to $(1-v)\one /2^{N}+v \rho$, where $v$ denotes the state's visibility. Given such a noise, we can examine the change of $N_k$ in a simple way since quantum expectations included in the inequalities are only rescaled by the visibility $v$. For example, in the $k=1$ case, the five-qubit state defined in (\ref{eq:1-poly}) yields no violation for $v\leq \sqrt{5/8}$, and one has to adjust $N_1$ to a higher value.

\begin{table}[h]
\caption{\label{tab:poly} Minimal $N_k$ needed to observe the $k$-polygamous behavior of the MABK inequalities.}
\vspace{0.3cm}
\begin{tabular}{c| c c c c c c c c c c c c }
\toprule
 $k$ \,& \, $1$ & $2$ & $3$ & $4$ & $5$ & $6$ & $7$ & $8$ & $9$ & $10$ & $11$ & $12$\\ 
\hline
$N_{k}$ & \, $5$ & $7$ & $10$ & $12$ & $15$ & $18$ & $21$ & $23$ & $26$ & $29$ & $32$ & $35$
\\
\bottomrule
\end{tabular}
\end{table}

\subsection{Maximal violation of the symmetrized MABK inequalities}

In the Bell polygamous scenario, all inequalities between $(N-k)$ parties can be violated simultaneously. The states provided in Eq.~\eqref{eq:max_poly} reach the maximum of the $(N-k)$-qubit MABK inequalities among the reduced states of the permutation invariant pure $N$-qubit states. However, each of these inequalities cannot be violated maximally, as this happens for the GHZ state~\cite{MERMIN, Ardehali_1992, Belinskii1993}. Here, we will show that the 
$N$-qubit inequality $M^{\mathrm{sym}}_{N,k}$ defined as the sum of MABK inequalities acting on all $\binom{N}{k}$ qubit subsystems $S$ of size $N-k$ and the optimal GHZ settings is violated maximally by the $k$-polygamous states provided in Eq.~\eqref{eq:max_poly}. More specifically, we shall prove that Eq.~\eqref{eq:max_poly} is the eigenvector of the corresponding Bell operator 
    \begin{align}\label{eq:Ank}
    M^{\text{sym}}_{N,k}:= \sum_{\substack{\forall S:\,\, \text{size}(S)=N-k}} (M_{N-k}) |_S \ot \one_k.
\end{align}
associated with its maximal eigenvalue. We will refer to the resulting inequality as the symmetrized MABK inequality. Its maximum quantum value for the given settings $Q^{\text{max}}_{N,k}$ is established through
the following Proposition.
\begin{proposition}\label{lem:A}
For $N/2>k$, the maximum expectation value of $M^{\text{sym}}_{N,k}$ over all $n$-qubit states $\ket{\psi}$,
\begin{align}
Q^{\text{max}}_{N,k}=\max_{\ket{\psi}}\bra{\psi} M^{\text{sym}}_{N,k}\ket{\psi}
\end{align}
is reached by $|\phi^k_N\rangle$ of Eq.~\eqref{eq:max_poly} and is given by
\begin{align}
Q^{\text{max}}_{N,k}=2^{N-k-1}\binom{N}{k}\binom{k}{[\,k/2\,]} \binom{N}{[\,k/2\,]} \binom{N}{k-[\,k/2\,]}\,. \nonumber
\end{align}
\end{proposition}

\begin{proof}
To start, it is convenient to rewrite Eq.~\eqref{eq:Ank} as
\begin{align}\label{eq:perminv}
    M^{\text{sym}}_{N,k} &= \frac{1}{k!(N-k)!} \sum_{P\in S_n}P(M_{N-k}\ot \one_k)P^\dagger \nonumber \\
    &= \frac{1}{k!(N-k)!} \text{sym}[M_{N-k}\ot \one_k]
    \,,
\end{align}
where $S_n$ is the permutation group acting on $n$-qubits and we define $\mathrm{sym}[O]:=\sum_{P\in S_n} P O P^\dagger$ as the symmetrisation of an operator $O$.
Note that $M^{\text{sym}}_{N,k}$ is a permutation-invariant operator and thus, it can be block-diagonalised when expressed in the spin basis~\cite{Moroder_2012}.
As a consequence, the maximisation from Eq.~\eqref{eq:Ank} is reduced to an optimisation over states of the following form
\begin{align}\label{eq:spin_state}
\ket{\psi_{j,\alpha}}=\sum^{j}_{m=-j} g_{j,m,\alpha} \ket{j,m,\alpha} 
\end{align}
where $\lbrace \ket{j,m,\alpha} \rbrace$ is the spin basis with $j= j_{\min}, j_{\min} + 1, \dots , N/2$ starting from $j_{\min}=0$ or $j_{\min}=1/2$ depending on whether $N$ is even or odd, respectively. Here $\alpha$ counts the degeneracy of the blocks with the same value of $j$ and $m=-j,\dots,j$.
For $j=N/2$, the block does not have a degeneracy, and the spin basis corresponds to Dicke states. Namely we have
\begin{align}
\ket{D^{i}_N} :=\ket{N/2,N/2-i,1} \,. 
\end{align}
Since all blocks with for given $j$ are equal, without loss of generality, one can fix $\alpha=1$. Thus,
\begin{align}\label{eq:spinopt}
\max_{\ket{\psi}}\bra{\psi}M^{\text{sym}}_{N,k}\ket{\psi} =\max_{j, \ket{\psi_{j,1}} }\bra{\psi_{j,1}}M^{\text{sym}}_{N,k}\ket{\psi_{j,1}}
\end{align}
where we maximise over all spin states with $\alpha=1$.

Since we have expressed our problem in the spin basis, we will also rewrite $M^{\text{sym}}_{N,k}$ in terms of spin matrices for convenience. First, note that $M_{N-k}$ can be written as~\cite{MERMIN}
\begin{align}
    (\mathcal{M}_{N-k}\ot \one_k)=2^{N-k-1}(\sigma^{\ot N-k }_-  + \sigma^{\ot N-k }_+ )\ot \one_k
\end{align}
with $\sigma_\pm$ being the ladder operators,
\begin{align}
\sigma_-=\begin{pmatrix} 0&&1 \\ 0&&0\end{pmatrix} \quad \text{and } \quad \sigma_+=\begin{pmatrix} 0&&0 \\ 1&&0\end{pmatrix} \,.
\end{align}
From this it follows that  $M^{\text{sym}}_{N,k}$ can be re-expressed as
\begin{align}\label{eq:pawelexp}
   \mathrm{sym} [M_{N-k}\ot \one_k]=k!(S^{N-k}_+ +S^{N-k}_-)
\end{align} 
with $S_{\pm}=\sum^N_i \sigma^{(i)}_{\pm}$ where $\sigma^{(i)}_{\pm}$ is a $N$-qubit operator acting non-trivially in qubit $j$ with $\sigma_{\pm}$.

This can be seen from the following argument. Note that 
\begin{align}\label{eq:sigma}
    \sym[\sigma^{\ot N-k }_- \ot \one_k]&=\sigma^{(1)}_-   \sym[\sigma_-^{\ot N-k-1} \ot \one_k]^{(2,\dots,n)}\nonumber\\&  +\sigma^{(2)}_-  \sym[\sigma_-^{\ot N-k-1} \ot \one_k]^{(1,3,\dots,n)}  \nonumber \\& 
    + \dots  \\&
    + \sigma^{(n)}_-   \sym[\sigma_-^{ \ot N-k-1} \ot \one_k]^{(1,2,\dots,n-1)} \,. \nonumber 
\end{align}
Where $\sym[O]^{S}$ is the symmetrisation acting non-trivially on qubits from the set $S$ with $O$ being an operator with support on these qubits.
Using the fact that $\sigma_-\sigma_-=0$, we can write the first term of the above sum as 
\begin{align}\label{eq:decomp}
&\frac{\sigma^{(1)}_-}{(k+1)} \Big((k+1)\sym[\sigma_-^{\ot N-k-1} \ot \one_k]^{(2,\dots,n)} \nonumber \\  &+ \sigma^{(1)}_-(N-k-1) \sym[\sigma_-^{\ot N-k-2} \ot \one_{k+1}]^{(2,\dots,n)} \Big)\,.  
\end{align}
However, this is nothing else than
\begin{align}
\frac{1}{(k+1)}\sigma_-^{(1)} \sym[\sigma_-^{\ot N-k-1} \ot \one_{k}] \,.
\end{align}
since $(k+1)$ and $(N-k-1)$ take into account the repeating elements that yield a sum over all $n$ permutations.
Similarly, one can do the same for all terms in Eq.~\eqref{eq:sigma} and obtain
\begin{align}
  \sym[\sigma^{\ot N-k }_- \ot \one_k] & = \frac{\sum^N_{i=1}\sigma_-^{(i)}}{(k+1)} \\&
  \cdot \sym[\sigma_-^{\ot N-k-1} \ot \one_{k+1}] \,. \nonumber
\end{align}
The application of the previous formula recursively gives
\begin{align}
  \sym[\sigma^{\ot N-k }_- \ot \one_k] &=   \frac{\left(\sum^N_{i=1}\sigma_-^{(i)}\right)^{N-k-1}}{(k+1)\dots (n-1)} \sym[\sigma_- \ot \one_{n-1}] \nonumber \\ 
  &= k!\left(\sum^N_{i=1}\sigma_-^{(i)}\right)^{N-k} \,.
\end{align} 
By repeating the procedure for $\sym[\sigma^{\ot N-k }_+ \ot \one_k]$, we obtain Eq.~\eqref{eq:pawelexp}.
Combining Eqs.~\eqref{eq:perminv} and ~\eqref{eq:pawelexp} results in
\begin{align}
M^{\text{sym}}_{N,k}=\frac{1}{(N-k)!} (S^{N-k}_+ +S^{N-k}_-) \,.
\end{align}

Finally, we will prove that the set of states which reach the maximum from Eq.~\eqref{eq:spinopt} can be restricted to states written in the Dicke basis. 
It is important to recall that $S_{\pm}$ acts on the spin basis as
\begin{align}
S_{\pm} \! \ket{j,m,\alpha} &=\sqrt{j(j+1)\!-\!m(m\pm 1)} \ket{j,m\pm 1,\alpha} \,. 
\end{align}
We use this to compute the expectation value of $M^{\text{sym}}_{N-k}$ for the general spin state as
\begin{align}\label{eq:dev}
&\bra{\psi_{j,1}} M^{\text{sym}}_{N,k} \ket{\psi_{j,1}} \nonumber \\&= \sum^{j}_{m,\,m'=-j}  \ \bra{j,m',1} M^{\text{sym}}_{N,k} \ket{j,m,1} c_{j,m'}c_{j,m}\nonumber \\ 
&\propto \prod^{N-k}_{i=1}\sqrt{j(j+1)-m(m+i)} c_{j,m+N+k}c_{j,m} \nonumber \\ &+ \prod^{N-k}_{i=1}\sqrt{j(j+1)-m(m- i)} c_{j,m+N-k}c_{j,m} \nonumber \\
&=f_+(j,m)+f_-(j,m),
\end{align}
where $c_{j,m} := g_{j,m,1}$ [see Eq.~\eqref{eq:spin_state}].

Note that $f_{\pm}(j,m)$ increases with $j$, and thus the maximum of Eq.~\eqref{eq:dev} is reached for $j=N/2$.
Thus,
\begin{align}
\max_{j, \ket{\psi_{j,1}} }\bra{\psi_{j,1}}M^{\text{sym}}_{N,k}\ket{\psi_{j,1}} = \max_{\ket{\text{PI}_N} }\bra{\text{PI}_N}M^{\text{sym}}_{N,k}\ket{\text{PI}_N}, \nonumber
\end{align}
with ${\ket{\psi_{N/2,1}}}=\ket{\text{PI}_N}$ as defined in Eq.~\eqref{eq:PI_state}.
Since $\bra{\text{PI}_N}M^{\text{sym}}_{N,k}\ket{\text{PI}_N}=\binom{N}{k}\bra{\text{PI}_N} (M_{N-k}\ot \one)\ket{\text{PI}_N}$, the maximization is equivalent to the one shown in Eq.~\eqref{eq:maxMABK} but, with an extra factor of $\binom{N}{k}$.
This ends the proof.
\end{proof}

Going back to Eq.~\eqref{eq:Ank}, the classical bound for the symmetrized MABK inequality is given as 
\begin{align}
L^{\text{sym}}_{N,k}= 2^{(N-k-1)/2}\binom{N}{k}
\end{align}
and as a consequence, the maximum violation factor reads
\begin{align}\label{eq:qoverc}
\frac{Q^{\max}_{N,k}}{L^{\mathrm{sym}}_{N,k}} = 2^{(N-k+1)/2}\binom{k}{[\,k/2\,]} \binom{N}{[\,k/2\,]} \binom{N}{k-[\,k/2\,]}\,.
\end{align}
This value is equal to the sum of maximal simultaneous violations of the MABK inequalities for all subsystems of size $N-k$. 
Thus, all results for $k$-polygamy states violating MABK inequalities can be extrapolated to symmetrized MABK. 
For example,
Table~\ref{tab:poly} shows for a given $k$, the minimum $N_k$ to obtain a violation of the symmetrized MABK.
The above result opens up a possibility for self-testing of all subsets of $N-k$ observers in an $N$-party communication scenario at once using the polygamous features~\cite{Panwar2023}, but we do not attempt to prove this here and leave these for future work on the subject.

\subsection{Optimal $N_2$ beyond MABK inequalities}

As mentioned before, MABK inequalities can yield general claims about the polygamous Bell correlations, but are not expected to give the optimal, i.e., the lowest $N_k$. Using the linear programming method introduced in~\cite{poly}, we establish the optimal $N_k$ for $k=2$. The corresponding $N_2$, we found that the generalized polygamy is possible starting from $N_2=6$, as compared to $N_2=7$ in the MABK case. The corresponding inequality is $\langle I_{ABCD} \rangle \leq 6$, where 
\begin{eqnarray}
I_{ABCD} &=& -2~{\rm Sym}[A_1B_2] - 2~{\rm Sym}[A_2 B_2] \\&–& A_1B_1C_1D_1 + {\rm Sym}[A_1B_1C_1D_2] \nonumber\\&-& {\rm Sym}[A_1B_1C_2D_2]- {\rm Sym}[A_1B_2C_2D_2] \nonumber \\&+& A_2B_2C_2D_2 \nonumber
\end{eqnarray}
and similarly for other partitions. Here, $A_{1/2}$ stands for the observable measured by the first party (for the other parties, we use the alphabetical order). We use a compact notation for symmetrizing over different observers
\begin{eqnarray}
{\rm Sym}[A_k B_l C_m D_n] = \sum_{\pi(k,l,m,n)} A_{k} B_{l} C_{m} D_{n},
\end{eqnarray}
where the sum is over all permutations of $(k,l,m,n)$, denoted as $\pi(k,l,m,n)$, assuming $A_0 = B_0 = C_0 =D_0 = 1$, e.g., 
${\rm Sym}[A_1 B_2]={\rm Sym}[A_1 B_2 C_0 D_0] =   A_1 B_2 C_0 D_0+A_2 B_1 C_0 D_0
+A_1 B_0 C_2 D_0+A_2 B_0 C_1 D_0
+A_1 B_0 C_0 D_2+A_2 B_0 C_0 D_1
+A_0 B_1 C_2 D_0+A_0 B_2 C_1 D_0 
+A_0 B_1 C_0 D_2+A_0 B_2 C_0 D_1
+A_0 B_0 C_1 D_2+A_1 B_2 C_2 D_1=
 A_1 B_2 +A_2 B_1 +A_1 C_2+A_2 C_1
+A_1 D_2+A_2 D_1+B_1 C_2 D_0+ B_2 C_1  
+B_1 C_0 D_2+ B_2 D_1+C_1 D_2+C_2 D_1$
being the permutations of $k=1, l=2,m=0,n=0$. Note that the symmetrization procedure used here differs from the one used in the previous section, where it applies to qubits.  
The optimal state yielding all four-party violations equal $6.271 > 6$ is given as
\begin{equation}
\ket{\phi}= \frac{1}{\sqrt{2}}\left(|\phi_{6}^{0}\rangle+|\phi_6^{4}\rangle\right).
\end{equation}
where $|\phi_6^{0}\rangle$ and $|\phi_6^{4}\rangle$ are defined in Eq.~\eqref{eq:max_poly}. The corresponding pairs of observables are the same for all parties and given by $A_i=\cos \varphi_{i} \sigma_x + \sin \varphi_i \sigma_y$ ($i=1,2$) with $\varphi_1 = 0.9047$ and $\varphi_2 = 1.9652$, where $\sigma_{i}$ denotes the Pauli matrices.

\section{Comparison of the polygamous violations with the standard GHZ states approach}
Let us now focus on the possible usage of the polygamous states and their comparison with the quantum violations of MABK inequalities achieved by the GHZ states. 

First, we compare the possible violations of the distinct MABK inequalities for the $k$-polygamous and the GHZ states. We will do so by analyzing the sum of squared MABK inequalities violation factors over all subsystems
\begin{align}\label{eq:nonlocality}
\mathcal{S}_\psi:=\sum_{\substack{\forall S:\,\text{size}(S)=N-k}} \frac{|\bra{\psi}(M_{N-k}) |_S \ot \one_k  \ket{\psi}|^2}{2^{(n-k-1)}}\,.
\end{align} 
Note that whenever $\mathcal{S}_{\psi}> {N \choose k}$, at least one of the inequalities has to be violated. The above expression has a similar form to the complementarity relations for Bell violations and sector lengths in entanglement detection~\cite{Zukowski_2002,Hassan_2008, Badziag_2008,Kurzynski2011, Tran_2015, Ketterer_2019, Wyderka_2020}. In order to obtain any violations in Eq.~\eqref{eq:nonlocality} using the GHZ states, one would usually send $|\text{GHZ}_{N-k}\rangle$ to a given partition $S$. In such a case, one gets
\begin{equation}
    \mathcal{S}_{\text{GHZ}_{N-k}}=2^{N-k-1}.
\end{equation}
Observing a higher number of violations would require multiple experiments aiming at testing different partitions $S_\psi$ in each of them. 
For the polygamous states, each of these violation factors is quite small compared to the ones achievable by the GHZ states. However, due to polygamy, all of them exceed 1 and contribute to the sum in $\mathcal{S}_{\psi}$. In fact, this feature allows the polygamous states to become advantageous over the GHZ states. 
From Proposition~\ref{lem:A}, we know that $|\phi^k_N\rangle$ is the state that reaches the maximum expectation value of $M^{\text{sym}}_{N-k}$. 
This can be seen from the fact that MABK inequalities are non-negative when the operator is written in the
$N$-qubit computational basis.
Thus, without loss of generality, we can assume that all expectation values for the Mermin-type operator are positive when considering the maximal value.
As a consequence, the optimal $k$-polygamous state will also give a maximum of $\mathcal{S}_\psi$.
The highest $\mathcal{S}_{\psi}$ is thus given as 
\begin{equation}
 \mathcal{S}_{\phi^k_N}=2^{-(N-k+1)}\binom{N}{k}A^{N,k}(\lceil k/2\rceil)^2,
\end{equation}
and corresponds to the states from Eq.~\eqref{eq:max_poly}.
One can see that for all $k\neq 0$ the polygamous state must be at least as good as the GHZ state of $N-k$ qubits. The explicit comparison shows that $\mathcal{S}_{\phi^k_N}$ and $\mathcal{S}_{\text{GHZ}_{N-k}}$ are exactly the same for $k=1$, while for $k>1$ the polygamous states outmatch the GHZ strategy, see Fig. \ref{fig:hyper_plot}. One could use this example to perform a certification of the non-classicality of the source on multiple subsystems at once or a nonlocality-based benchmark of a quantum device utilizing the fact of binomially many violations at once.

\begin{figure}[h]
    \centering
    \includegraphics[width=0.48\textwidth]{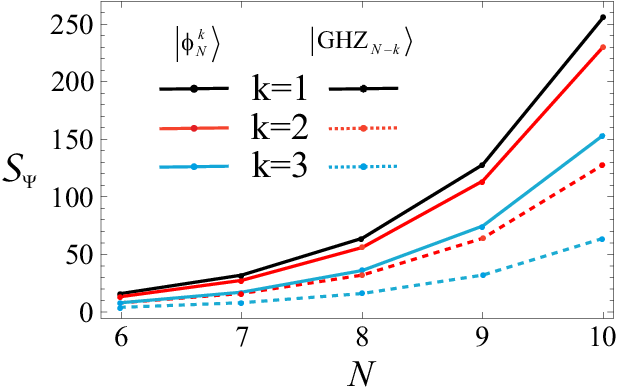}
    \caption{Comparison of the sum of squared violation factors of the $(N-k)$-qubit MABK inequalities for the $\ket{\text{GHZ}_{N-k}}$ (dashed lines) and the optimal $k$-polygamous states  $|\phi_N^{k} \rangle$  (solid lines). Each color corresponds to a different value of $k=1,2,3$. Despite the fact that the GHZ states give rise to the maximal violation of a given MABK inequality, the amount of small violations over all subsystems makes the polygamous states achieve higher values of $\mathcal{S_{\psi}}$.}
    \label{fig:hyper_plot}
\end{figure}

\section{$K$-Hyper-polygamy}
Exploring the concept of $k$-polygamy led us to the observation of an even more exciting phenomenon. Namely, we were able to find states for which the $k$-polygamy holds for multiple subsystem sizes at once. Given an upper bound $K$, on the number of discarded parties $k$, we will refer to this structure as $K$-hyper-polygamy.
Consider the first non-trivial example of $K=2$. 
As in the previous section, we will work with the $N$-qubit permutation invariant states violating $(N-k)$-party MABK inequalities. 
Now, we want to construct a state which is simultaneously $2$- and $1$-polygamous. One possible choice is to take a superposition of the three states that reach the maximum violation of $\mathcal{M}_{N-k}$ [see Eq.~\eqref{eq:max_poly}]: two for $k=1$ and one for $k=2$. This leads to the following state
\begin{align}\label{eq:sum_hyper}
 \frac{1}{\sqrt{10}}\left({\ket{D_N^{0}} + 2{\ket{D_N^{1}}+2\ket{D_N^{N-1}}}
 + \ket{D_N^{N}}}\right), 
\end{align}
which can be generalized  using a single parameter $\alpha \in [0,1]$ via
\begin{align}\label{eq:hyperpoly}
 \ket{\Psi^2_N}&=\alpha\frac{{\ket{D_N^{1}}+\ket{D_N^{N-1}}}}{\sqrt{2}} \nonumber \\ &+ \sqrt{1-\alpha^2} \frac{\ket{D_N^{0}} + \ket{D_N^{N}}}{\sqrt{2}} .
\end{align}
Computing the expectations of $(N-1)$- and $(N-2)$-party Mermin-type operators on the above state using Eq.~\eqref{eq:merm_simp}, one gets
\begin{align} 
    \mathcal{M}_{N-1}/L_{N-1}&= \frac{2^{N/2}}{\sqrt{N}}\alpha\sqrt{1-\alpha^2} \,,  \nonumber\\
    \mathcal{M}_{N-2}/L_{N-2}&= \frac{2^{(N-1)/2}}{N} \alpha^2 \,.
\end{align}
where again $L_{N-k}$ is the local bound. Now, we want to check if one can choose such $\alpha$ that both violation factors exceed $1$.
Indeed, such a task is possible whenever
\begin{align}\label{eq:limithyper}
\alpha^2 \in \left(\, \frac{N}{2^{(N-1)/2}}, \quad \frac{1}{2}+\frac{1}{2^{N/2}}\sqrt{2^{N-2} - N} \right) \,,
\end{align}
where the lower part is determined by $\mathcal{M}_{N-2}/L_{N-k}$ and the upper part by $\mathcal{M}_{N-1}/L_{N-k}$.
This region is non-zero, i.e., the upper part is larger than the lower part, from $N=7$. Hence, violation of all two-setting $(N-1)$- and $(N-2)$-qubit Bell inequalities is possible when $N>7$. Complementarily, it is not possible to obtain a $2$-hyper-polygamous state for $N< 7$ using the state construction from Eq.~\eqref{eq:hyperpoly}. This is consistent with our previous findings on $k$-polygamy.
Eq.~\eqref{eq:limithyper} also shows that for large enough $N$, one obtains $2$-hyper-polygamous states for almost all $\alpha$. 
In a similar manner, we can superpose all possible $k$-polygamous states that reach the maximal violation for $0 \leq k\leq K$. By normalizing the states and grouping the terms as in Eq.~\eqref{eq:hyperpoly}, we obtain
\begin{align}\label{eq:gen_hyperpoly}
\ket{\Psi^K_{N}}=\sum^{\lceil{K/2\rceil}}_{i=0} \alpha_i\frac{{\ket{D_N^{i}}+\ket{D_N^{N-i}}}}{\sqrt{2}} \,,
\end{align}
with $\alpha_i \in  [0,1]$.
Although it might be the case that for a given $N$, one can find a $K$-hyper-polygamous state written as Eq.~\eqref{eq:gen2_hyperpoly}, we are interested in finding the minimal $N_K$ for which there exists a $K$-hyper-polygamous state leading to the nonlocal expectations of the examined Mermin-type operator.

This problem can be formulated with semidefinite programming (SDP) by rewriting 
Eq.~\eqref{eq:merm_simply} for mixed states
\begin{align}
 \mathcal{M}_{N-k}(\varrho) = \sum_{s=0}^k A^{N,k}(s) \, \varrho_{s,N-k+s} \,.
\end{align}
The SDP then reads 
\begin{align}\label{eq:SDP}
    \quad & \max_\varrho \sum^K_{k=1} \mathcal{M}_{N-k}(\varrho)/L_{N-k} \nonumber \,,\\
   \text{s.t.} \quad &  \mathcal{M}_{N-k}(\varrho)/L_{N-k} \geq 1 \quad \text{for} \quad K=1,\dots, k \nonumber \,,\\
   \quad & \tr(\varrho) =1\,.  \end{align}
The result of the SDP is likely to be close to a pure state because the maximization is over a sum of expectation values and the constraints only impose a lower bound on them (with the exception of the normalization).
In fact, for all of the tested cases, the state was always pure. 
The SDP from Eq.~\eqref{eq:SDP} allowed us to find the minimal $N_K$ such that there exists a $K$-hyper-polygamous state.
These results are shown in Table~\ref{tab:hyper}. 
\begin{table}[h]
\caption{\label{tab:hyper} Minimal $N_K$ to obtain a $K$-hyper-polygamous behavior for MABK inequalities and the optimal GHZ state settings.}
\vspace{0.3cm}
\begin{tabular}{c| c c c c c c c c c c c }
\toprule
 $K$ \,& \, $2$ & $3$ & $4$ & $5$ & $6$ & $7$ & $8$ & $9$ & $10$ & $11$ & $12$\\ 
\hline
$N_{K}$ \, & \,$7$ & $10$ & $13$ & $15$ & $18$ & $21$ & $24$ & $26$ & $29$ & $32$ & $35$
\\
\bottomrule
\end{tabular}
\end{table} 
From the SDP solutions, we conjecture that the states written as
\begin{align}\label{eq:gen2_hyperpoly}
\ket{\Phi^K_{N}}=\sum^{K}_{i=0} \beta_i\frac{{\ket{D_N^{i}}+\ket{D_N^{N-i}}}}{\sqrt{2}}\,, 
\end{align}
are $K$-hyper-polygamous for a particular choice of $\beta_i \in [0,1]$.
For the cases from Table~\ref{tab:hyper}, this conjecture is fulfilled.
Note that Eq.~\eqref{eq:gen_hyperpoly} is a specific case of Eq.~\eqref{eq:gen2_hyperpoly}.
In fact, the SDP shows that a $N_K$ qubit state written as Eq.~\eqref{eq:gen2_hyperpoly} cannot always be a $K$-hyper-polygamous.
For the cases where $K=5,\,6,\, 9,\, 10$, the state $\ket{\Psi^K_{N}}$ is $K$-hyper-polygamous when $N$ is at least $16,\, 19, \, 28, \,30$, respectively. 
However, Table~\ref{tab:hyper} shows that the minimum is achieved when $N_K=15,\, 18, \, 27,\, 29$.
It is important to note that $N_k$ for polygamy is not always the same as $N_K$ for hyper-polygamy (see Tables~\ref{tab:poly} and \ref{tab:hyper}). Indeed, we see that
$N_k \neq N_K$ for $K=k=4$ and $8$.

Here, it is important to note that the examined Mermin-type operator $\mathcal{M}_{N-k}$ arising from the MABK inequalities requires different measurement settings for each $k$ in order to observe the violations. These settings are the optimal GHZ state settings and are given as $O_1=\cos\theta_{N-k}\sigma_x+\sin \theta_{N-k} \sigma_y$ and $O_2=\cos(\theta_{N-k}-\pi/2)\sigma_x+\sin (\theta_{N-k}-\pi/2) \sigma_y$ for each observer, where $\theta_{N-k}=(N-k-1)\pi/(4[N-k])$. This means that the hyperpolygamous Bell violations cannot be directly observed using the MABK inequalities at the same time. However, it can be verified that the standard Mermin inequalities~\cite{MERMIN} with $\sigma_x,\sigma_y$ settings can. This is due to the fact that they give rise to the same Bell operator, but with a different local bound for even $N-k$, i.e. $L^{\text{even}}_{N-k}=2^{(N-k)/2}$.
This makes the adjustment of $N_K$ necessary as shown in Table~\ref{tab:hyper_mermin}, but the general effect holds, including the conjecture shown in Eq.~\eqref{eq:gen2_hyperpoly}. 
Nevertheless, the family of states in  Eq.~\eqref{eq:hyperpoly} is no longer $2$-hyper-polygamous for $N=7$, and one requires going up to $N=9$, to see a non-zero region in the corresponding modification of Eq.~\eqref{eq:limithyper}.

Note that the $K$-hyper-polygamous states are Bell nonlocal analogies to the concept of $m$-resistant states~\cite{Quinta_2019} -- the states that stay entangled after a loss of arbitrary $m$ parties, but become separable when more particles are traced out. In fact, the $1$-polygamous state (\ref{eq:1-poly}) lies in the same plane as the $1$-resistant state of three qubits.

\begin{table}[h]
\caption{\label{tab:hyper_mermin} Minimal $N_K$ to obtain a $K$-hyper-polygamous behavior for Mermin inequalities. In bold, the differences with respect to Table~\ref{tab:hyper} are highlighted.}
\vspace{0.3cm}
\begin{tabular}{c| c c c c c c c c c c c }
\toprule
 $K$ \,& \, $2$ & $3$ & $4$ & $5$ & $6$ & $7$ & $8$ & $9$ & $10$ & $11$ & $12$\\ 
\hline
$N_{K}$ \, & \,$7$ & $10$ & $13$ & $\boldsymbol{16}$ & $\boldsymbol{19}$ & $\boldsymbol{22}$ & $\boldsymbol{25}$ & $\boldsymbol{28}$ & $29$ & $32$ & $35$
\\
\bottomrule
\end{tabular}
\end{table} 

\section{Conclusions}
In this work, we explored the polygamous structure in the violations of Bell inequalities for multipartite quantum systems. Using the symmetries of the MABK inequalities and the corresponding Bell operators, we constructed a family of permutation-invariant states violating all possible MABK inequalities of a given size at once. Namely, we proved that for any $k>0$ there always exists a system size $N_k$ for which all MABK inequalities on $(N_k-k)$-qubits are violated with a single $ N_k$-qubit state. Furthermore, we show that the minimal $N_k$ for which this is possible scales approximately linearly with $k$, and the number of discarded subsystems in the Bell test can be as high as $\approx N/3$. Thus, even after losing a third of the whole system, one can still observe Bell nonlocal correlations.

These findings demonstrate that quantum correlations in multipartite systems are significantly more abundant than previously recognized, with the amount of shared nonlocality scaling extensively with the number of possible subsystem partitions.  This unique property of polygamous states makes it possible to improve the benchmarks of nonlocality in quantum devices at many subsystems at once. Moreover, $k$-polygamous states stay Bell nonlocal after a loss of an arbitrary set of $k$ parties, and this property could become especially useful in communication scenarios. We specifically show that an $N$-qubit Bell inequality composed of the symmetrization of a $(N-k)$-party Mermin-type operator reaches its maximum value for the introduced $k$-polygamous states, which are optimal in scenarios with positive integer $k<N/2$. Moreover, they outperform the violation strategies based on $|\text{GHZ}_{N-k}\rangle$ for $k > 1$ (at $k=1$ both strategies give the same result). This observation gives future perspectives for device-independent self-testing of multiple groups of nodes in a quantum network at once, providing a step forward towards scalable quantum technologies.

We further explored the concept of $K$-hyper-polygamy, where a single quantum state exhibits $k$-polygamy for all $1\leq k \leq K$, allowing for an even greater number of simultaneous violations of Bell inequalities. Using SDP, we determined the minimal $N_K$ for $K\leq 12$. For example, the 2-hyper-polygamy exists already in a system of $N=7$ qubits. This demonstrates the possibility of observing 28 simultaneous Bell violations in all $6$- and $5$-qubit subsystems using a single seven-qubit state.

\section*{Acknowledgements}

We thank Marcin Paw{\l}owski and Tam\'as V\'ertesi for helpful discussions.
The authors are supported by the National Science Centre (NCN, Poland) within the OPUS project (Grant No. 2024/53/B/ST2/04103). PC acknowledges the support of the Foundation for Polish Science (FNP) within the START program. This work is partially carried out under IRA Programme, project no. FENG.02.01-IP.05-0006/23, financed by the FENG program 2021-2027, Priority FENG.02, Measure FENG.02.01., with the support of the FNP. 

\bibliography{ref}

@article{poly,
author = {Paweł Cieśliński  and Lukas Knips  and Mateusz Kowalczyk  and Wiesław Laskowski  and Tomasz Paterek  and Tamás Vértesi  and Harald Weinfurter },
title = {Unmasking the polygamous nature of quantum nonlocality},
journal = {Proceedings of the National Academy of Sciences},
volume = {121},
number = {44},
pages = {e2404455121},
year = {2024},
doi = {10.1073/pnas.2404455121}
}

@article{Cui_2025,
  title = {Monogamy of nonlocal games},
  author = {Cui, David and Mehta, Arthur and Rochette, Denis},
  journal = {Phys. Rev. Res.},
  volume = {7},
  issue = {3},
  pages = {L032003},
  numpages = {6},
  year = {2025},
  month = {Jul},
  publisher = {American Physical Society},
  doi = {10.1103/1m7q-rgrh},
  url = {https://link.aps.org/doi/10.1103/1m7q-rgrh}
}

@article{Quinta_2019,
  title = {Cut-resistant links and multipartite entanglement resistant to particle loss},
  author = {Quinta, G. M. and Andr\'e, Rui and Burchardt, Adam and \ifmmode \dot{Z}\else \.{Z}\fi{}yczkowski, Karol},
  journal = {Phys. Rev. A},
  volume = {100},
  issue = {6},
  pages = {062329},
  numpages = {14},
  year = {2019},
  month = {Dec},
  publisher = {American Physical Society},
  doi = {10.1103/PhysRevA.100.062329},
  url = {https://link.aps.org/doi/10.1103/PhysRevA.100.062329}
}

@article{Pawlowski_review_2025,
title = {The future of secure communications: Device independence in quantum key distribution},
journal = {Physics Reports},
volume = {1149},
pages = {1-97},
year = {2025},
issn = {0370-1573},
doi = {https://doi.org/10.1016/j.physrep.2025.09.006},
url = {https://www.sciencedirect.com/science/article/pii/S0370157325002558},
author = {Seyed Arash Ghoreishi and Giovanni Scala and Renato Renner and Letícia Lira Tacca and Jan Bouda and Stephen Patrick Walborn and Marcin Pawłowski}
}

@inproceedings{Mayers_1998,
  series = {SFCS-98},
  title = {Quantum cryptography with imperfect apparatus},
  url = {http://dx.doi.org/10.1109/SFCS.1998.743501},
  DOI = {10.1109/sfcs.1998.743501},
  booktitle = {Proceedings 39th Annual Symposium on Foundations of Computer Science (Cat. No.98CB36280)},
  publisher = {IEEE Comput. Soc},
  author = {Mayers,  D. and Yao,  A.},
  year = {1998},  
  pages = {503–509},
  collection = {SFCS-98}
}

@article{Barrett_2005,
  title = {No {S}ignaling and {Q}uantum {K}ey {D}istribution},
  author = {Barrett, Jonathan and Hardy, Lucien and Kent, Adrian},
  journal = {Phys. Rev. Lett.},
  volume = {95},
  issue = {1},
  pages = {010503},
  numpages = {4},
  year = {2005},
  month = {Jun},
  publisher = {American Physical Society},
  doi = {10.1103/PhysRevLett.95.010503},
  url = {https://link.aps.org/doi/10.1103/PhysRevLett.95.010503}
}

@article{Acin_2007,
  title = {Device-{I}ndependent {S}ecurity of {Q}uantum {C}ryptography against {C}ollective {A}ttacks},
  author = {Ac\'{\i}n, Antonio and Brunner, Nicolas and Gisin, Nicolas and Massar, Serge and Pironio, Stefano and Scarani, Valerio},
  journal = {Phys. Rev. Lett.},
  volume = {98},
  issue = {23},
  pages = {230501},
  numpages = {4},
  year = {2007},
  month = {Jun},
  publisher = {American Physical Society},
  doi = {10.1103/PhysRevLett.98.230501},
  url = {https://link.aps.org/doi/10.1103/PhysRevLett.98.230501}
}

@article{Pironio_2009,
  title = {Device-independent quantum key distribution secure against collective attacks},
  volume = {11},
  ISSN = {1367-2630},
  url = {http://dx.doi.org/10.1088/1367-2630/11/4/045021},
  DOI = {10.1088/1367-2630/11/4/045021},
  number = {4},
  journal = {New Journal of Physics},
  publisher = {IOP Publishing},
  author = {Pironio,  Stefano and Acín,  Antonio and Brunner,  Nicolas and Gisin,  Nicolas and Massar,  Serge and Scarani,  Valerio},
  year = {2009},
  month = apr,
  pages = {045021}
}

@article{Pironio_2010,
  title = {Random numbers certified by {B}ell’s theorem},
  volume = {464},
  ISSN = {1476-4687},
  url = {http://dx.doi.org/10.1038/nature09008},
  DOI = {10.1038/nature09008},
  number = {7291},
  journal = {Nature},
  publisher = {Springer Science and Business Media LLC},
  author = {Pironio,  S. and Acín,  A. and Massar,  S. and de la Giroday,  A. Boyer and Matsukevich,  D. N. and Maunz,  P. and Olmschenk,  S. and Hayes,  D. and Luo,  L. and Manning,  T. A. and Monroe,  C.},
  year = {2010},
  month = apr,
  pages = {1021–1024}
}

@misc{Colbeck_2011,
      title={Quantum {A}nd {R}elativistic {P}rotocols {F}or {S}ecure {M}ulti-{P}arty {C}omputation}, 
      author={Roger Colbeck},
      note={arXiv:0911.3814 [quant-ph]},
      year={2011}
   }

@article{Colbeck_2011_2,
  title = {Private randomness expansion with untrusted devices},
  volume = {44},
  ISSN = {1751-8121},
  url = {http://dx.doi.org/10.1088/1751-8113/44/9/095305},
  DOI = {10.1088/1751-8113/44/9/095305},
  number = {9},
  journal = {Journal of Physics A: Mathematical and Theoretical},
  publisher = {IOP Publishing},
  author = {Colbeck,  Roger and Kent,  Adrian},
  year = {2011},
  month = feb,
  pages = {095305}
}

@article{ArnonFriedman_2018,
  title = {Practical device-independent quantum cryptography via entropy accumulation},
  author = {Arnon-Friedman,  Rotem and Dupuis,  Frédéric and Fawzi,  Omar and Renner,  Renato and Vidick,  Thomas},
journal = {Nature Communications},  
volume = {9},
  ISSN = {2041-1723},
  number = {1},
  publisher = {Springer Science and Business Media LLC},
  year = {2018},
  month = {Jan},
url={https://www.nature.com/articles/s41467-017-02307-4}
}

@article{Cleve_1997,
  title = {Substituting quantum entanglement for communication},
  author = {Cleve, Richard and Buhrman, Harry},
  journal = {Phys. Rev. A},
  volume = {56},
  issue = {2},
  pages = {1201--1204},
  numpages = {0},
  year = {1997},
  month = {Aug},
  publisher = {American Physical Society},
  doi = {10.1103/PhysRevA.56.1201},
  url = {https://link.aps.org/doi/10.1103/PhysRevA.56.1201}
}

@article{Brukner_2002,
  title = {Quantum {C}ommunication {C}omplexity {P}rotocol with {T}wo {E}ntangled {Q}utrits},
  author = {Brukner,  {\v C}aslav and \ifmmode \dot{Z}\else \.{Z}\fi{}ukowski, Marek and Zeilinger, Anton},
  journal = {Phys. Rev. Lett.},
  volume = {89},
  issue = {19},
  pages = {197901},
  numpages = {4},
  year = {2002},
  month = {Oct},
  publisher = {American Physical Society},
  doi = {10.1103/PhysRevLett.89.197901},
  url = {https://link.aps.org/doi/10.1103/PhysRevLett.89.197901}
}

@article{Moreno_2020,
  title = {Device-independent secret sharing and a stronger form of {B}ell nonlocality},
  author = {Moreno, M. G. M. and Brito, Samura\'{\i} and Nery, Ranieri V. and Chaves, Rafael},
  journal = {Phys. Rev. A},
  volume = {101},
  issue = {5},
  pages = {052339},
  numpages = {11},
  year = {2020},
  month = {May},
  publisher = {American Physical Society},
  doi = {10.1103/PhysRevA.101.052339},
  url = {https://link.aps.org/doi/10.1103/PhysRevA.101.052339}
}

@phdthesis{Toner2007PhD,
  title={Quantifying quantum nonlocality},
  author={Toner, Benjamin Francis},
  year={2007},
  school={California Institute of Technology},
 url={https://catalogue.nla.gov.au/catalog/7124052}
}

@article{Toner2009mono,
  title={Monogamy of non-local quantum correlations},
  author={Toner, Ben},
  journal={Proc.R.Soc.A},
  volume={465},
  number={2101},
  pages={59--69},
  year={2009},
  publisher={The Royal Society London},
    doi={
https://doi.org/10.1098/rspa.2008.0149
}
}

@misc{TVmono2006,
  title={Monogamy of {Bell} correlations and {Tsirelson's} bound},
  author={Toner, Benjamin and Verstraete, Frank},
  note = {arXiv:0611001 [quant-ph]},
    year={2006},
}

@article{Wiesniak2021,
  title = {Symmetrized persistency of Bell correlations for Dicke states and GHZ-based mixtures},
  volume = {11},
  ISSN = {2045-2322},
  url = {http://dx.doi.org/10.1038/s41598-021-93786-5},
  DOI = {10.1038/s41598-021-93786-5},
  number = {1},
  journal = {Scientific Reports},
  publisher = {Springer Science and Business Media LLC},
  author = {Wieśniak,  Marcin},
  year = {2021},
  month = jul 
}

@article{Bell1964, 
title={On the {E}instein {P}odolsky {R}osen paradox}, 
volume={1}, 
ISSN={0554-128X},
number={3}, 
journal={Phys. Phys. Fiz.}, 
publisher={American Physical Society (APS)}, 
author={Bell, J. S.}, 
year={1964}, 
doi = {https://doi.org/10.1103 PhysicsPhysiqueFizika.1.195},
month=nov, 
pages={195–200}
}

@article{MERMIN,
  title = {Extreme quantum entanglement in a superposition of macroscopically distinct states},
  author = {Mermin, N. David},
  journal = {Phys. Rev. Lett.},
  volume = {65},
  issue = {15},
  pages = {1838--1840},
  numpages = {0},
  year = {1990},
  month = {Oct},
  publisher = {American Physical Society},
doi = {https://doi.org/10.1103/PhysRevLett.65.1838}
}

@article{SG2001,
  title = {Quantum Communication between $\mathit{N}$ Partners and {Bell's} Inequalities},
  author = {Scarani, Valerio and Gisin, Nicolas},
  journal = {Phys. Rev. Lett.},
  volume = {87},
  issue = {11},
  pages = {117901},
  numpages = {4},
  year = {2001},
  month = {Aug},
  publisher = {American Physical Society},
doi={https://doi.org/10.1103/PhysRevLett.87.117901}
}

@article{Brunner_2014,
  title = {Bell nonlocality},
  author = {Brunner, Nicolas and Cavalcanti, Daniel and Pironio, Stefano and Scarani, Valerio and Wehner, Stephanie},
  journal = {Rev. Mod. Phys.},
  volume = {86},
  issue = {2},
  pages = {419--478},
  numpages = {60},
  year = {2014},
  month = {Apr},
  publisher = {American Physical Society},  
  doi ={https://doi.org/10.1103/RevModPhys.86.419}
}

@article{Collins2004,
year = {2004},
month = {jan},
publisher = {},
volume = {37},
number = {5},
pages = {1775},
author = {Daniel Collins and  Nicolas Gisin},
title = {A relevant two qubit {Bell} inequality inequivalent to the {CHSH} inequality},
journal = {J. Phys. A},
doi ={10.1088/0305-4470/37/5/021}
}

@article{Kurzynski2011,
  title = {Correlation Complementarity Yields {Bell} Monogamy Relations},
  author = {Kurzy\'{n}ski, P. and Paterek, T. and Ramanathan, R. and Laskowski, W. and Kaszlikowski, D.},
  journal = {Phys. Rev. Lett.},
  volume = {106},
  issue = {18},
  pages = {180402},
  numpages = {4},
  year = {2011},
  month = {May},
  publisher = {American Physical Society},
  doi = {https://doi.org/10.1103/PhysRevLett.106.180402}
}

@article{Ram_2018,
  title = {Trade-offs in multiparty Bell-inequality violations in qubit networks},
  author = {Ramanathan, Ravishankar and Mironowicz, Piotr},
  journal = {Phys. Rev. A},
  volume = {98},
  issue = {2},
  pages = {022133},
  numpages = {15},
  year = {2018},
  month = {Aug},
  publisher = {American Physical Society},
doi = {https://doi.org/10.1103/PhysRevA.98.022133}}

@article{Augusiak2014,
  title = {Elemental and tight monogamy relations in nonsignaling theories},
  author = {Augusiak, R. and Demianowicz, M. and Paw\l{}owski, M. and Tura, J. and Ac\'{\i}n, A.},
  journal = {Phys. Rev. A},
  volume = {90},
  issue = {5},
  pages = {052323},
  numpages = {10},
  year = {2014},
  month = {Nov},
  publisher = {American Physical Society},
doi={https://doi.org/10.1103/PhysRevA.90.052323}
}

@article{Ramanathan2014,
  title = {Strong Monogamies of No-Signaling Violations for Bipartite Correlation {Bell} Inequalities},
  author = {Ramanathan, Ravishankar and Horodecki, Pawe\l{}},
  journal = {Phys. Rev. Lett.},
  volume = {113},
  issue = {21},
  pages = {210403},
  numpages = {5},
  year = {2014},
  month = {Nov},
  publisher = {American Physical Society},
doi={https://doi.org/10.1103/PhysRevLett.113.210403}
}

@article{Tran2018,
  title = {{Bell} monogamy relations in arbitrary qubit networks},
  author = {Tran, M. C. and Ramanathan, R. and McKague, M. and Kaszlikowski, D. and Paterek, T.},
  journal = {Phys. Rev. A},
  volume = {98},
  issue = {5},
  pages = {052325},
  numpages = {8},
  year = {2018},
  month = {Nov},
  publisher = {American Physical Society},
  doi = {https://doi.org/10.1103/PhysRevA.98.052325}
}

@article{Augusiak2017,
  title = {Simple and tight monogamy relations for a class of {Bell} inequalities},
  author = {Augusiak, Remigiusz},
  journal = {Phys. Rev. A},
  volume = {95},
  issue = {1},
  pages = {012113},
  numpages = {8},
  year = {2017},
  month = {Jan},
  publisher = {American Physical Society},
doi = {https://doi.org/10.1103/PhysRevA.95.012113}
}

@article{Belinskii1993,
year = {1993},
month = {aug},
publisher = {},
volume = {36},
number = {8},
pages = {653},
author = {A V Belinskiĭ and  D N Klyshko},
title = {Interference of light and {Bell's} theorem},
journal = {Phys.-Uspekhi},
doi = {https://doi.org/10.1070/PU1993v036n08ABEH002299}
}

@article{Ardehali_1992,
  title = {Bell inequalities with a magnitude of violation that grows exponentially with the number of particles},
  author = {Ardehali, M.},
  journal = {Phys. Rev. A},
  volume = {46},
  issue = {9},
  pages = {5375--5378},
  numpages = {0},
  year = {1992},
  month = {Nov},
  publisher = {American Physical Society},
doi = {https://doi.org/10.1103/PhysRevA.46.5375}
}

@article{Scarani_2001,
year = {2001},
month = {jul},
publisher = {},
volume = {34},
number = {30},
pages = {6043},
author = {Valerio Scarani and  Nicolas Gisin},
title = {{Spectral decomposition of Bell's operators for qubits}},
journal = {Journal of Physics A: Mathematical and General},
doi= {10.1088/0305-4470/34/30/314}
}

@article{Nagata_2006,
  title = {Bell inequality with an arbitrary number of settings and its applications},
  author = {Nagata, Koji and Laskowski, Wies\l{}aw and Paterek, Tomasz},
  journal = {Phys. Rev. A},
  volume = {74},
  issue = {6},
  pages = {062109},
  numpages = {8},
  year = {2006},
  month = {Dec},
  publisher = {American Physical Society},
doi = { https://doi.org/10.1103/PhysRevA.74.062109}
}

@article{Moroder_2012,
  title = {Permutationally invariant state reconstruction},
  volume = {14},
  ISSN = {1367-2630},
  url = {http://dx.doi.org/10.1088/1367-2630/14/10/105001},
  DOI = {10.1088/1367-2630/14/10/105001},
  number = {10},
  journal = {New Journal of Physics},
  publisher = {IOP Publishing},
  author = {Moroder,  Tobias and Hyllus,  Philipp and Tóth,  Géza and Schwemmer,  Christian and Niggebaum,  Alexander and Gaile,  Stefanie and G\"{u}hne,  Otfried and Weinfurter,  Harald},
  year = {2012},
  month = oct,
  pages = {105001}
}

@article{Hassan_2008,
  title = {Experimentally accessible geometric measure for entanglement in $\mathit{N}$-qubit pure states},
  author = {Hassan, Ali Saif M. and Joag, Pramod S.},
  journal = {Phys. Rev. A},
  volume = {77},
  issue = {6},
  pages = {062334},
  numpages = {11},
  year = {2008},
  month = {Jun},
  publisher = {American Physical Society},
  doi = {10.1103/PhysRevA.77.062334},
  url = {https://link.aps.org/doi/10.1103/PhysRevA.77.062334}
}

@article{Badziag_2008,
  title = {Experimentally Friendly Geometrical Criteria for Entanglement},
  author = {Badzia\ifmmode \mbox{\c{}}\else \c{}\fi{}g, Piotr and Brukner, {\v C}aslav and Laskowski, Wies\l{}aw and Paterek, Tomasz and \ifmmode \dot{Z}\else \.{Z}\fi{}ukowski, Marek},
  journal = {Phys. Rev. Lett.},
  volume = {100},
  issue = {14},
  pages = {140403},
  numpages = {4},
  year = {2008},
  month = {Apr},
  publisher = {American Physical Society},
  doi = {10.1103/PhysRevLett.100.140403},
  url = {https://link.aps.org/doi/10.1103/PhysRevLett.100.140403}
}

@article{Ketterer_2019,
  title = {Characterizing Multipartite Entanglement with Moments of Random Correlations},
  author = {Ketterer, Andreas and Wyderka, Nikolai and G\"uhne, Otfried},
  journal = {Phys. Rev. Lett.},
  volume = {122},
  issue = {12},
  pages = {120505},
  numpages = {6},
  year = {2019},
  month = {Mar},
  publisher = {American Physical Society},
  doi = {10.1103/PhysRevLett.122.120505},
  url = {https://link.aps.org/doi/10.1103/PhysRevLett.122.120505}
}

@article{Tran_2015,
  title = {Quantum entanglement from random measurements},
  author = {Tran, Minh Cong and Daki\ifmmode \acute{c}\else \'{c}\fi{}, Borivoje and Arnault, Francois and Laskowski, Wies\l{}aw and Paterek, Tomasz},
  journal = {Phys. Rev. A},
  volume = {92},
  issue = {5},
  pages = {050301},
  numpages = {7},
  year = {2015},
  month = {Nov},
  publisher = {American Physical Society},
  doi = {10.1103/PhysRevA.92.050301},
  url = {https://link.aps.org/doi/10.1103/PhysRevA.92.050301}
}

@article{Wyderka_2020,
  title = {Characterizing quantum states via sector lengths},
  volume = {53},
  ISSN = {1751-8121},
  url = {http://dx.doi.org/10.1088/1751-8121/ab7f0a},
  DOI = {10.1088/1751-8121/ab7f0a},
  number = {34},
  journal = {Journal of Physics A: Mathematical and Theoretical},
  publisher = {IOP Publishing},
  author = {Wyderka,  N and G\"{u}hne,  O},
  year = {2020},
  month = jul,
  pages = {345302}
}

@article{Zukowski_2002,
  title = {Bell's Theorem for General N-Qubit States},
  author = {\ifmmode \dot{Z}\else \.{Z}\fi{}ukowski, Marek and Brukner,  {\v C}aslav},
  journal = {Phys. Rev. Lett.},
  volume = {88},
  issue = {21},
  pages = {210401},
  numpages = {4},
  year = {2002},
  month = {May},
  publisher = {American Physical Society},
  doi = {10.1103/PhysRevLett.88.210401},
  url = {https://link.aps.org/doi/10.1103/PhysRevLett.88.210401}
}

@article{Brukner_2004,
  title = {Bell's Inequalities and Quantum Communication Complexity},
  author = {Brukner,  {\v C}aslav and \ifmmode \dot{Z}\else \.{Z}\fi{}ukowski, Marek and Pan, Jian-Wei and Zeilinger, Anton},
  journal = {Phys. Rev. Lett.},
  volume = {92},
  issue = {12},
  pages = {127901},
  numpages = {4},
  year = {2004},
  month = {Mar},
  publisher = {American Physical Society},
  doi = {10.1103/PhysRevLett.92.127901},
  url = {https://link.aps.org/doi/10.1103/PhysRevLett.92.127901}
}

@article{Murta2020,
  title = {Quantum Conference Key Agreement: A Review},
  volume = {3},
  ISSN = {2511-9044},
  url = {http://dx.doi.org/10.1002/qute.202000025},
  DOI = {10.1002/qute.202000025},
  number = {11},
  journal = {Advanced Quantum Technologies},
  publisher = {Wiley},
  author = {Murta,  Gláucia and Grasselli,  Federico and Kampermann,  Hermann and Bruß,  Dagmar},
  year = {2020},
  month = sep 
}

@article{Panwar2023,
  title = {An elegant scheme of self-testing for multipartite Bell inequalities},
  volume = {9},
  ISSN = {2056-6387},
  url = {http://dx.doi.org/10.1038/s41534-023-00735-3},
  DOI = {10.1038/s41534-023-00735-3},
  number = {1},
  journal = {npj Quantum Information},
  publisher = {Springer Science and Business Media LLC},
  author = {Panwar,  Ekta and Pandya,  Palash and Wieśniak,  Marcin},
  year = {2023},
  month = jul 
}

\appendix

\end{document}